\documentclass[lettersize,journal]{IEEEtran}
\IEEEoverridecommandlockouts

\usepackage{cite}
\usepackage{amsmath,amssymb,amsfonts, bm}
\usepackage{algorithmic}
\usepackage[ruled,lined,linesnumbered]{algorithm2e}
\usepackage{graphicx}
\usepackage{tikz}
\usepackage{textcomp}
\usepackage{indentfirst}
\usepackage{color, colortbl}
\usepackage[english]{babel}
\usepackage{wrapfig}
\usepackage{enumitem}
\usepackage{subcaption}
\usepackage{xfrac}
\usepackage{diagbox}
\usepackage{array}
\usepackage{multirow}
\usepackage{arydshln}
\usepackage{tcolorbox}
\usepackage{mathtools}
\usepackage{hyperref}
\newtheorem{theorem}{Theorem}
\newtheorem{lemma}{Lemma}

\newtheorem{proof}{Proof}

\usepackage[T1]{fontenc}

\usepackage{marvosym}
\usetikzlibrary{arrows.meta, matrix, positioning}
\usetikzlibrary{calc,shapes,shapes.geometric,patterns}
\tikzset{
    queue element/.style={
        draw,very thin,rounded corners,
        fill=yellow!30,
        font=\sffamily\footnotesize
    },
    >={[scale=0.8]Triangle}
}

\tikzset{
    node/.style={circle, draw, fill=lightgray!30, minimum size=5mm, font=\small},
    big/.style={circle, draw, fill=lightgray!30, minimum size=7mm, font=\small},
    edge/.style={->, >=Stealth, thin},
}

\def\BibTeX{{\rm B\kern-.05em{\sc i\kern-.025em b}\kern-.08em
    T\kern-.1667em\lower.7ex\hbox{E}\kern-.125emX}}
\begin{document}

\title{Congestion Reduction in EV Charger Placement Using Traffic Equilibrium Models
\thanks{
$\dagger$: Contributed equally. 
$\ddagger$: Corresponding author.
Yasin Sonmez, Can Kizilkale and Murat Arcak are with the Department of Electrical Engineering and Computer Sciences at the University of California, Berkeley. Alex Kurzhanskiy is with the California Partners for Advanced Transportation Technology
(PATH) at the University of California, Berkeley. Semih Kara is with Gensyn AI (work conducted while at the University of Maryland, College Park). Nuno C. Martins is with the Department of Electrical and Computer Engineering and the Institute for Systems Research at the University of Maryland, College Park.
This project was sponsored in part by the NSF CPS grant NSF CNS-2135791 and by the California Statewide Research Transportation Program. 
}
}

\author{\IEEEauthorblockN{ Semih Kara$^{1,\dagger}$ \hspace{2cm}
Yasin Sonmez$^{2,\dagger}$ \hspace{2cm}
Can Kizilkale$^{2,\ddagger}$ \\ \vspace{2em}
Alex Kurzhanskiy$^{2}$ \hspace{2cm} Nuno C. Martins$^{3}$ \hspace{2cm} Murat Arcak$^{2}$}\\ \vspace{2em}
\IEEEauthorblockA{
$^1$Gensyn AI}\\
skara@terpmail.umd.edu\\
\IEEEauthorblockA{
$^2$University of California, Berkeley}\\
 \{yasin\_sonmez, cankizilkale, akurzhan, arcak\}@berkeley.edu\\
\IEEEauthorblockA{
$^3$University of Maryland, College Park}\\
nmartins@umd.edu}

\author{Semih Kara$^{\dagger}, \quad$ 
Yasin Sonmez$^{\dagger}, \quad$ 
Can Kizilkale$^{\ddagger}$,\\
Alex Kurzhanskiy, \quad
Nuno C. Martins, \quad
Murat Arcak

}

\maketitle

\begin{abstract}
Growing EV adoption can worsen traffic conditions if chargers are sited without regard to their impact on congestion. We study how to strategically place EV chargers to reduce congestion using two equilibrium models: one based on congestion games and one based on an atomic queueing simulation. We apply both models within a scalable greedy station‐placement algorithm. Experiments show that this greedy scheme yields optimal or near‐optimal congestion outcomes in realistic networks, even though global optimality is not guaranteed as we show with a counterexample. We also show that the queueing-based approach yields more realistic results than the congestion‐game model, and we present a unified methodology that calibrates congestion delays from queue simulation and solves equilibrium in link‐space.\\

\end{abstract}

\section{Introduction}

\IEEEPARstart{E}{lectric} vehicle (EV) growth is proceeding at an unprecedented rate \cite{iea}, with global targets aiming for 100 million EVs by 2030 \cite{unitednations}. 
This requires a momentous expansion of charging infrastructure; California alone projects a need for nearly 1.2 million public chargers for passenger EVs and 157,000 for commercial vehicles by 2030 \cite{CAC}.
However, such a dramatic enlargement faces significant challenges \cite{10092390}. Major investments are needed for the power distribution system, as current grid capacity severely limits charger placement~\cite{liander_net, pge_grid,ahmed,zhu}. Furthermore, zoning approvals and permits are time-consuming. For the foreseeable future, the number of chargers is unlikely to match the increase in EVs, potentially leading to queues, traffic disruption, and significant travel time from detours to limited charger locations.

Motivated by this emerging supply-demand imbalance, in this paper, we focus on the problem of EV charger location placement to reduce traffic congestion, both by mitigating congestion caused directly by EVs and by strategically using charging locations and charging needs to actively shape traffic flows at equilibrium. Prior work on charger network design has addressed related challenges from multiple perspectives. Early foundational work integrated travel demand and accessibility into siting decisions \cite{Huang2016TRD}, while data-driven approaches leveraging mobility traces offer complementary insights \cite{Vazifeh2019TRA}. Of particular relevance to our study, equilibrium-aware charger placement has received growing attention, including bilevel network design models with elastic demand and station congestion under user equilibrium \cite{Huang2020TRD}, MPEC formulations incorporating driving range constraints \cite{He2015TRC,He2018TRC}, and extensions to pricing mechanisms \cite{Mirheli2023TITS}. Congestion game frameworks have been applied to both optimal placement \cite{Sonmez2024Optimal} and joint optimization of placement and pricing strategies \cite{Gasnier2025PricingPlacement}. Additional work has addressed extreme fast-charging stations \cite{Fang2025TITS}, stochastic user equilibrium for mixed EV/gasoline traffic \cite{Zeng2025TRE}, and integrated models that jointly consider traffic assignment and power grid requirements \cite{Mao2021TITS}.

We follow two approaches to compute equilibrium flows. The first uses a congestion game framework to enable theoretical analysis. A congestion game \cite{rosenthal1973class} is a game-theoretic model where multiple users choose among shared resources, and the cost (delay) of each resource increases with the number of users choosing it (i.e., its congestion). This is a useful abstraction; however, for traffic, these models are usually too idealized. They omit important physical and operational constraints, such as spillback effects and vehicle dynamics at intersections, and they ignore practical network details (e.g., number of lanes, speed limits, geometric features).
The second approach is based on a mesoscopic queueing-based traffic simulator, which is more realistic and better suited for accurate simulation studies. However, mesoscopic models are more difficult to analyze theoretically due to the time-dependent/dynamic nature of the problem. Examples of related work in this space include dynamic equilibrium (Nash flow over time) in fluid queueing networks and constructive methods for computing such equilibria \cite{Koch2011}, as well as continuous-time analogues of key static network flow results \cite{KOCH2014580}. Reference \cite{Cominetti2015Dynamic} models congestion in fluid queueing networks and establishes the existence, uniqueness, and global stability of Wardrop dynamic equilibrium. In addition, under a necessary capacity condition, dynamic equilibria in fluid queueing networks are shown to reach a predictable steady state in finite time \cite{Cominetti2022LongTerm}. 

In this paper, we study how to reduce congestion by strategically placing EV chargers. We develop two equilibrium models, one based on congestion games and the other on a more realistic atomic queueing simulation.  We apply both within a highly scalable greedy station placement algorithm (Algorithm \ref{alg:greedyEV}). For the greedy algorithm, we adopt a similar methodology to \cite{Sonmez2024Optimal}, which proposes a computationally tractable two–layer optimization framework: the inner layer computes equilibrium flows for a candidate set of charger locations while the outer layer selects the charger locations.

Our experiments show that, under both equilibrium formulations, greedy placement produces optimal or near–optimal congestion outcomes relative to the model being used. We also show with an example that global optimality is not guaranteed. We then compare the congestion-game abstraction with the queue-based model and demonstrate that the queue-based approach yields significantly different outcomes.

Our contributions are: (i) explicit modeling of traffic congestion and equilibrium routing behavior through both congestion game theory and queue-based simulation; (ii) a tractable link-space congestion game formulation that avoids a decision route set while implicitly considering all feasible paths through flow conservation constraints; (iii) a unified methodology that calibrates congestion game delay functions from queue simulations, solves equilibrium in link-space, and recovers interpretable route flows; and (iv) empirical validation showing that greedy placement achieves optimality in realistic networks.
The queue-based model in item (i), as well as items (ii), (iii), and (iv) are entirely new and were not addressed in the conference paper \cite{Sonmez2024Optimal}.

The paper is organized as follows: We first introduce both the congestion game and the queue–based equilibrium models. We then introduce the Optimal EV Placement problem and describe the greedy station placement algorithm and methodology. We conclude with experimental results.

\section{Congestion Game model} \label{sec:ev-as-congestion-game}
\noindent
In the EV placement model using a congestion game framework, we evaluate the total equilibrium delay for each candidate location and select the one that minimizes it. For the traffic problems, the road network is the resource, and a driver's cost is travel time minus the charging benefit. With non-atomic flow, link delays depend on total usage. Chargers are modeled as additional network links with their own delays.

In the congestion game model, the road network is modeled as a directed graph $G=(V,E)$ where $V$ represents nodes and $E$ represents links representing physical road segments. The travel time on a link $l \in E$ is given by $d_l(x_l)$, where $x_l \in \mathbb{R}^+$ denotes the flow on $l$ (since the total flow is bounded, the link flows are also bounded). The delay function $d_l$ is assumed to be non-decreasing with respect to $x_l$.
Given a path $s$ (a sequence of distinct vertices with each consecutive pair connected by an edge), the total delay is the sum of the delays on its links: \[
\sum_{l \in s} d_l(x_l).
\]
Let $V_c \subseteq V$ be the set of candidate nodes where an EV station can be placed. A charger at node $v \in V_c$ is represented as an additional self-loop $(v,v)$, indicating that if a route includes $(v,v)$, the EV stops at $v$ to charge. Each path may include at most one such charging link. The time spent at a charging station is modeled by a non-decreasing function $d_l(x_{(v,v)})$ and added to the total delay of any path that includes it. The set of all self-loop charging links is denoted by $E_{charger}$.

\textit{Agent types.} The population is partitioned into types $(i,t)$. $i \in \{1, \ldots , N\}$ denotes origin destination pairs $O_i,D_i$. We will denote the feasible set of routes for agent type $(i,t)$ between $O_i$ and $D_i$ as $S_{(i,t)}$, and $t \in \{F_1, F_2, F_3\}$ specifies charging requirements (let $T=\{1, \ldots, N\} \times \{F_1, F_2, F_3\}$ represent the type set):
\begin{itemize}
    \item $F_1$: Agents that do not require charging (e.g., gasoline vehicles or EVs with sufficient charge). Their feasible paths exclude self-links $(v,v)$.
    \item $F_2$: Agents that must charge once en route. Their feasible paths include exactly one self-link $(v,v)$ corresponding to a charging station.
    \item $F_3$: Agents that may benefit from charging but can also complete the trip without it. The benefit depends on waiting time, cost, and detour; those with no benefit are also included in this group.
\end{itemize}
\textit{Agent Cost Functions.} While agent types in $(i, F_1)$ and $(i, F_2)$ are similar in the sense that they try to pick the route with minimum delay, the agents in the third set decide whether the additional benefit from charging outweighs a longer route delay. We represent the cost for agent type $(i,t)$ when taking route $s_{(i,t)} \in S_{(i,t)}$ as $u^{(i,t)}(s_i)$. For agents of type $(i,t)$ where $t \in \{F_1,F_2\}$ the cost of a route $s_{(i,t)}$ is $$u^{(i,t)}(s_{(i,t)})=\sum_{l\in s_{(i,t)}}d_l(x_l).$$ A path of an agent in type $(i, F_1)$ will not include a self-directed charging link, while, for type in $(i, F_2)$, the path $s_{(i,t)}$ will contain exactly one such link. For agents of type $(i, F_3)$, if the route $s_{(i,t)}$ does not include a charging link then the cost function will be the same, $\sum_{l\in s_{(i,t)}}d_l(x_l)$; otherwise, it will have an additional term $c_i$ representing the benefit of charging for agent sub-type $i$, hence the cost will be $\sum_{l\in s_{(i,F_3)}}d_l(x_l) -c_i$ and $$u^{(i,F_3)}(s_{(i,F_3)})=\sum_{l\in s_{(i,F_3)}}(d_l(x_l) - \mathbf{1}(l \in E_{\text{charger}})c_i)$$ where $\mathbf{1}(l \in E_{\text{charger}})$ represents an indicator function.

\subsection{Nash Equilibrium for the congestion game framework}
\noindent
The Nash equilibrium concept is central concept in game theory. At such an equilibrium, no agent (in our context, a driver) can obtain a better outcome by unilaterally changing their decision while others keep theirs fixed \cite{Nash-noncoop-games}. Nash Equilibrium in a continuous-flow congestion game with an infinite number of infinitesimal players is also known as the Wardrop Equilibrium. In a congestion game, this means that a strategy profile $s^*$ constitutes a Nash equilibrium (NE) if
\[
u^{(i,t)}(s^*_{(i,t)}) \leq u^{(i,t)}(s_{(i,t)})
\]
for each agent type $(i,t)$ and feasible paths $s_{(i,t)} \in S_{(i,t)}$. The Nash equilibrium framework has been widely applied in traffic analysis \cite{rosenthal-nash-eq,bell-casir-traffic-nash-eq,meunier-wagner-eq-res-for-dyn-congestion-games}.

Under certain conditions, agents can learn to play a Nash equilibrium by following simple heuristics that align with realistic human behavior. For instance, in a potential congestion game, if agents iteratively switch to routes with lower delays, the system gradually moves toward a Nash equilibrium \cite{monderer-shapley-potential-games, sandholm-potential-games}. Therefore, Nash equilibria are often used to characterize steady operating points in traffic systems, even when agents are not fully rational. Motivated by this reasoning, we also make extensive use of Nash equilibria in this paper.

To find the equilibrium points of our congestion game model, we  use the following potential function:
\begin{equation}
\label{eq:ExtendedPotential}
    \sum_{l \in E \cup E_{\text{charger}} \ } \int_{0}^{\sum_{i,t} x^{(i,t)}_l} d_l(x)\,dx - \sum_{l \in E_{\text{charger}}}  \sum_{i} c_i x^{(i,F_3)}_l  
\end{equation}

This is an extension of the well known Beckmann potential and the optima of the following formulation are the Nash equilibria.
We minimize (\ref{eq:ExtendedPotential}) with respect to flow constraints. For each agent type $(i,t)$, the total flow over all feasible paths must equal the demand $q^{(i,t)}$:
\begin{align}
\label{eq:EPOPT}
   \min_{\{x^{(i,t)}_{s_{(i,t)}}\}} & \sum_{l \in E \cup E_{\text{charger}} \ } \int_{0}^{\sum_{i,t} x^{(i,t)}_l} d_l(x)\,dx - \sum_{l \in E_{\text{charger}}}  \sum_{i} c_i x^{(i,F_3)}_l \notag \\
   \text{s.t.} \quad & \sum_{s_{(i,t)} \in S_{(i,t)}} x^{(i,t)}_{s_{(i,t)}} = q^{(i,t)}, \quad \forall (i,t) \in T     
\end{align}
Here, with an abuse of notation, we represent the flow for agent of type $(i,t)$ on route $s_{(i,t)} \in S_{(i,t)}$ as $x^{(i,t)}_{s_{(i,t)}}$.

The following theorem uses the convexity of this extended potential to ascertain the correspondence between the minimizers and Nash equilibria:

\begin{theorem}
\label{prop:mainprop}
    The extended potential function (\ref{eq:ExtendedPotential}) is convex, and every minimizer $s^*$ of (\ref{eq:EPOPT}) is a Nash Equilibrium.
\end{theorem}
\begin{proof}
Given that the delay functions $d_l$ are non-decreasing, the summation over their  integrals, $\sum_{l \in E \cup E_{\text{charger}} \ } \int_{0}^{\sum_{i,t} x^{(i,t)}_l} d_l(x)\,dx$, is convex. Since $\sum_{l \in E_{\text{charger}}}  \sum_{i} c_i x^{(i,F_3)}_l$ is linear, (\ref{eq:ExtendedPotential}) is convex.
For some agent type $(i,t)$ let $s_{(i,t)} \in S_{(i,t)}$ improves over $s^*_{(i,t)}$ (path that optimizes the potential). If $t$ is $F_1$ or $F_2$ then $\sum_{l\in s_{(i,t)}} d_l(x_l) < \sum_{l\in s^*_{(i,t)}} d_l(x_l)$, if $t$ is $F_3$ then $\sum_{l\in s_{(i,t)}}(d_l(x_l) - \mathbf{1}(l \in E_{\text{charger}})c_i) < \sum_{l\in s^*_{(i,t)}}(d_l(x_l) - \mathbf{1}(l \in E_{\text{charger}})c_i)$ . But then for small enough $\delta$ increasing flows by $\delta$ in $s_{(i,t)}$ while decreasing by the same amount in $s^*_{(i,t)}$ we can decrease the potential function, which is a contradiction. Hence, every minimizer of \ref{eq:EPOPT} is a Nash Equilibrium.
\end{proof}

\begin{lemma}
    There exists an equilibrium for the congestion game.
\end{lemma}

\begin{proof}
Given the flows being bounded, \ref{eq:EPOPT} always has a minimizer. From Proposition \ref{prop:mainprop} all the minimizers are NE, which concludes the proof.     
\end{proof}

Since (\ref{eq:EPOPT}) can be solved as a convex optimization problem, this allows us to find NE points efficiently.

\subsection{Nash Equilibrium over link flows}
\noindent
Since the agent types induce linear constraints on the associated flows, the problem is formulated as a convex program. The formulation is expressed in terms of link flows, with distinct flow variables defined for each agent type and charger combination, thereby ensuring that the number of decision variables remains tractable.

\subsubsection*{Route-link formulation equivalence}
To establish the equivalence between the link-based and route-based formulations, we draw on the \emph{flow decomposition theorem} (see Ford and Fulkerson, \emph{Flows in Networks} \cite{FordFulkerson1962}). For any origin--destination pair \((o_i,d_i)\) and agent type \(t\), and in the absence of additional charging constraints, it is well known that any feasible edge flow in a single-commodity network can be decomposed into a collection of path flows. This decomposition result implies that a path-based formulation can always be rewritten as an edge-based formulation in which flows satisfy nodal conservation constraints and edge capacity bounds.  

When charging requirements are introduced, the same reasoning can be extended by refining the definition of agent types. Specifically, a type \((i,t)\) can be partitioned into subtypes according to the charging link that must be traversed. For each charging link \(c\), we introduce two subtypes: \((i,t,c_+)\), which represents the flow from the origin \(o_i\) to the head of link \(c\), and \((i,t,c_-)\), which represents the flow from the tail of \(c\) to the destination \(d_i\). Each of these subflows can in turn be expressed using nodal flow conservation and capacity constraints. To ensure that these subflows jointly represent the behavior of the original type \((i,t)\), we impose coupling constraints: the inflow \((i,t,c_+)\) and outflow \((i,t,c_-)\) associated with a charging link must be equal, and the sum of such flows across all charging links must match the total demand for type \((i,t)\).

As defined earlier, let $G = (V, E)$ denote the road network, where the cardinalities of $|V|$, $|E|$, and $|E_{\text{charger}}|$ are given by $n_v$, $n_e$, and $n_s$, respectively.

\subsubsection{Demand and Vehicle Types}
\noindent
There are $N$ origin-destination (OD) pairs, denoted as $\{(o_i, d_i)\}_{i=1}^N$. For each OD pair, we partition demand into three vehicle types:
\begin{equation*}
q^{(i,1)}: \text{non-charging,} \ \
q^{(i,2)}: \text{charging,} \ \
q^{(i,3)}: \text{may charge}
\end{equation*}
\noindent
Type 1 vehicles never charge en route, Type 2 always charge, and Type 3 may choose to charge depending on network conditions. For charging types, we further split the demand according to which charging link is used:
\[
q^{(i,t,c)}: \text{demand from OD $i$, type $t$, charging at station $c$}
\]

\subsubsection{Node Demand Vector}
\noindent
Let $s$ denote the node corresponding to charger $c$. We introduce a node demand vector for each OD pair $i$, type $t$, and route segment:
\begin{align}
y_v^{(i,t,c_+)} &=
\begin{cases}
    q^{(i,t,c)}, & v = o_i \\
    -q^{(i,t,c)}, & v = s \\
    0, & \text{otherwise}
\end{cases} \\
y_v^{(i,t,c_-)} &=
\begin{cases}
    q^{(i,t,c)}, & v = s \\
    -q^{(i,t,c)}, & v = d_i \\
    0, & \text{otherwise}
\end{cases} \\
y_v^{(i,t,nc)} &=
\begin{cases}
    q^{(i,t)} - \sum_c q^{(i,t,c)}, & v = o_i \\
    -q^{(i,t)} + \sum_c q^{(i,t,c)}, & v = d_i \\
    0, & \text{otherwise}
\end{cases}
\end{align}
These vectors encode how much demand enters and exits each node for each segment of the journey.

\subsubsection{Route--link mapping}
\noindent
Let $\mathcal{R}$ be a library of candidate routes (including charging and non-charging compositions) with cardinality $M$. Define the routing matrix $\mathbf{R} \in \{0,1\}^{n_e \times M}$ by
\[
\mathbf{R}_{\ell r} = \begin{cases}
1, & \text{if route } r \text{ uses link } \ell,\\
0, & \text{otherwise.}
\end{cases}
\]
For any nonnegative route-flow vector $\mathbf{f} \in \mathbb{R}^M_{\ge 0}$, the induced (non self-loop) link flows satisfy
\begin{equation}
\label{eq:xRf}
\mathbf{x} = \mathbf{R}\,\mathbf{f}.
\end{equation}

\subsubsection{Flow Variables and Conservation}
\noindent
To describe traffic movement, we define the following flow variables:
\begin{itemize}
    \item $x_l^{(i,t,c_+)}$: flow on link $l$ for vehicles of OD $i$, type $t$, traveling from $o_i$ to charger $c$ (before charging).
    \item $x_l^{(i,t,c_-)}$: flow on $l$ for vehicles from OD $i$, type $t$, traveling from charger $c$ to $d_i$ (after charging).
    \item $x_l^{(i,t,nc)}$: flow on $l$ for vehicles of OD $i$, type $t$, that never use a charging station.
\end{itemize}

Each $x_l^{(i,t,c)}$ is the sum of flows approaching and departing a charger:
\begin{align}
x_l^{(i,t,c)} = x_l^{(i,t,c_+)} + x_l^{(i,t,c_-)} \label{eq:splitflow}
\end{align}

The total flow on a link $l$ (excluding charger self-loops) sums over all OD pairs, vehicle types, and charger choices:
\begin{align}
x_l = \sum_{i=1}^N \sum_{t=1}^3 \left[ x_l^{(i,t,nc)} + \sum_{c \in E_{\text{charger}}} x_l^{(i,t,c)} \right]
\end{align}

For charger links $c$, the total flow passing through is:
\begin{align}
\hat{x}_c = \sum_{i=1}^N \sum_{t=1}^3 q^{(i,t,c)}
\end{align}

To impose flow conservation, we use the (reduced) incidence matrix $A \in \mathbb{R}^{n_v \times n_e}$, which encodes how links connect nodes:
\begin{equation}
A_{v,l} =
\begin{cases}
    1, & \text{if link } l \text{ starts at node } v \\
    -1, & \text{if link } l \text{ ends at node } v \\
    0, & \text{otherwise}
\end{cases}
\end{equation}
\noindent
Self-loop charger links are excluded from $A$ as they do not affect flow between distinct nodes.

\subsubsection{Optimization Problem Formulation}
\noindent
We denote the following link-space optimization as Problem~(CP). The goal is to allocate OD link flows and charging choices to minimize the Beckmann potential, accounting for both regular and charging links. The Beckmann potential function is defined as:
\begin{equation}
\Phi(\mathbf{x}, \hat{\mathbf{x}}) = \sum_{l \in E} \int_0^{x_l} d_l(\xi) d\xi + \sum_{c \in E_{\text{charger}}} \int_0^{\hat{x}_c} d_c(\xi) d\xi
\end{equation}
Then the optimization problem is given by:
\begin{align}
\tag{CP}\label{prob:CP}
\min_{\mathbf{x},\,\mathbf{\hat{x}},\,q} \quad & \Phi(\mathbf{x}, \hat{\mathbf{x}}) -\sum_{i=1}^{N} c_i \mathbf{1}^T\mathbf{\hat{x}}^{(i,3)}  \notag \\
\text{s.t.} \quad
& \mathbf{x}^{(i,t,c)} \geq 0,\ \mathbf{x}^{(i,t,nc)} \geq 0,\ \mathbf{\hat{x}} \geq 0,\ q^{(i,t,c)} \geq 0  \tag{\textit{Non-negativity}}\\
& A\mathbf{x}^{(i,t,c_+)} = \mathbf{y}^{(i,t,c_+)}, A\mathbf{x}^{(i,t,c_-)} = \mathbf{y}^{(i,t,c_-)},\ \forall i, t, c \tag{\textit{Flow conservation via charger $c$}}\\
& A\mathbf{x}^{(i,t,nc)} = \mathbf{y}^{(i,t,nc)},\ \forall i, t  \tag{\textit{Flow conservation without charging}}\\
& \sum_c q^{(i,t,c)} \leq q^{(i,t)},\ \forall i, t  \tag{\textit{Demand partitioning 1}}\\
& q^{(i,1,c)} = 0,\quad \sum_c q^{(i,2,c)} = q^{(i,2)},\ \forall i  \tag{\textit{Demand partitioning 2}}\\
& \hat{x}_c = \sum_{i=1}^{N} \sum_{t=1}^{3} q^{(i,t,c)},\ \forall c \tag{\textit{Total charger flow}} \\
& \mathbf{x} = \sum_{i=1}^N \sum_{t=1}^3 \left[ \mathbf{x}^{(i,t,nc)} + \sum_{c \in E_{\text{charger}}} \mathbf{x}^{(i,t,c)} \right] \tag{\textit{Total link flow}}
\end{align}
\noindent
The total delay experienced by all users is given by:
\begin{equation}
c_d(\mathbf{x}^*, \hat{\mathbf{x}}^*) = \sum_{l \in E} x_l^* \cdot d_l(x_l^*) + \sum_{c \in E_{\text{charger}}} \hat{x}_c^* \cdot d_c(\hat{x}_c^*)
\end{equation}

This represents the sum of flow times delay for each link, where $c_d(\mathbf{x}^*, \hat{\mathbf{x}}^*)$ denotes the total delay at the minimizer flows $(\mathbf{x}^*, \hat{\mathbf{x}}^*)$ of the above problem. The optimization objective in Problem~(CP) is equivalent to the objective function in equation~\eqref{eq:EPOPT}. Both formulations minimize the Beckmann potential minus the charging benefits for F$_3$ agents, ensuring that their minimizers correspond to Nash equilibria of the congestion game.

For easy reference, Table~\ref{tab:notation} summarizes the key notation used throughout this section.

\begin{table}[h]
\centering
\caption{Key Notation for Congestion Game Model}\label{tab:notation}
\begin{tabular}{ll}
\hline
\textbf{Symbol} & \textbf{Meaning} \\
\hline
$G=(V,E)$ & Road network: nodes and directed links \\
$V_c \subseteq V$ & Candidate charger locations \\
$E_{\text{charger}}$ & Self-loop links representing chargers \\
$N$ & Number of origin-destination (OD) pairs \\
$(i,t)$ & Agent type: $i \in \{1,\ldots,N\}$ (OD), $t \in \{F_1,F_2,F_3\}$ \\
$O_i, D_i$ & Origin and destination for OD pair $i$ \\
$S_{(i,t)}$ & Set of feasible routes from $O_i$ to $D_i$ \\
$q^{(i,t)}$ & Demand (flow) of agent type $(i,t)$ \\
$x^{(i,t)}_l$ & Flow on link $l$ from agents of type $(i,t)$ \\
$x^{(i,t)}_{s_{(i,t)}}$ & Flow on route $s_{(i,t)} \in S_{(i,t)}$ for agent type $(i,t)$ \\
$x_l$ & Total flow on link $l$: $\sum_{i,t} x^{(i,t)}_l$ \\
$d_l(x_l)$ & Delay (travel time) function for link $l$ \\
$c_i$ & Charging benefit for OD pair $i$, type $F_3$ agents \\
$u^{(i,t)}(s_{(i,t)})$ & Cost for agent type $(i,t)$ using route $s_{(i,t)}$ \\
\hline
\end{tabular}
\end{table}

\pagebreak
\section{Queue-Based Model}
\label{sec:queue-based-model}
\noindent
Congestion games provide a tractable mathematical framework for analyzing traffic. However, these models are usually not realistic representations. For the sake of tractability, they typically omit important physical and operational constraints, such as spillback effects and vehicle dynamics at intersections. They also do not account for many practical details of real road networks, such as the exact number of lanes, speed limits, or geometric features that influence congestion and travel times. To complement the theoretical analysis and account for real-world traffic dynamics, we also test our results using a mesoscopic queuing-based traffic simulator as a second model.

We implement this queuing-based simulation tool by extending a software package from the UC Berkeley PATH Program \cite{Sim_Package_Git}. The original version of the tool was designed to analyze traffic during wildfire evacuations and has been applied successfully in several studies \cite{Sim_Package_Report,zhao2019agent,zhao2021developing}. For our work, we extended it to include EV models and charging stations.

The tool simulates traffic at the level of individual vehicles, capturing important dynamics such as spillback and vehicle movements at intersections. Although it does not model the fine-grained microscopic details of simulators like SUMO \cite{guastella2025trafficmodelingsumotutorial}, it provides important practical advantages: it runs efficiently, requires less detailed data, and allows EV charging behavior to be incorporated easily, making it well-suited for our study. It also supports the use of OpenStreetMap data, allowing us to test our results on realistic roadway networks.

The simulation is organized around four components: a link model, a node model, a charging station model, and an agent model. At each timestep, a central controller updates the state of every vehicle by running these models in sequence. Each vehicle has a predefined origin, destination, and departure time. The link and node models represent the roads and intersections in the network, with their geospatial attributes taken from OpenStreetMap. The link model organizes vehicles in a spatial queue (Figure~\ref{fig:link}), while the node model simulates vehicle behavior at intersections.

\tikzset{treetop/.style = {decoration={random steps, segment length=0.4mm},decorate},trunk/.style = {decoration={random steps, segment length=2mm, amplitude=0.2mm},decorate}}

\tikzset{
   my redcar/.pic={
        \shade[top color=red, bottom color=white, shading angle={135}]
        [draw=black,fill=red!20,rounded corners=1.2ex,very thick] (0,0) -- ++(0.2,1.3) -- ++(4,0) --  ++(0.7,-1.3) -- ++(-4.9,0) -- cycle;
        \draw[very thick, rounded corners=0.5ex,fill=black!20!blue!20!white,thick]  (0.7,1.3) -- ++(0.25,1.2) -- ++(2.5,0) -- ++(0.25,-1.2) -- cycle;
        \draw[thick]  (2.2,1.3) -- (2.2,2.5);
        \draw[draw=black,fill=gray!50,thick] (1,0) circle (0.5);
        \draw[draw=black,fill=gray!50,thick] (3.6,0) circle (0.5);
    }
}
\tikzset{
   my greencar/.pic={
        \shade[top color=green, bottom color=white, shading angle={135}]
        [draw=black,fill=red!20,rounded corners=1.2ex,very thick] (0,0) -- ++(0.2,1.3) -- ++(4,0) --  ++(0.7,-1.3) -- ++(-4.9,0) -- cycle;
        \draw[very thick, rounded corners=0.5ex,fill=black!20!blue!20!white,thick]  (0.7,1.3) -- ++(0.25,1.2) -- ++(2.5,0) -- ++(0.25,-1.2) -- cycle;
        \draw[thick]  (2.2,1.3) -- (2.2,2.5);
        \draw[draw=black,fill=gray!50,thick] (1,0) circle (0.5);
        \draw[draw=black,fill=gray!50,thick] (3.6,0) circle (0.5);
    }
}

\begin{figure}[h]
\centering
\begin{tikzpicture}[every pic/.style={scale=0.3}]
\shade[gray!50] (0,0) rectangle (8.5,-0.6);
\pic at (1,0) {my greencar};
\pic at (5.3,0) {my redcar};
\pic at (7,0) {my redcar};

\draw (0,-0.8) to (0,-1.2);
\draw (0,-1) to (5.25,-1);
\draw (5.25,-0.8) to (5.25,-1.2);

\draw (5.3,-0.8) to (5.3,-1.2);
\draw (5.3,-1) to (8.5,-1);
\draw (8.5,-0.8) to (8.5,-1.2);

\node[align=center] at (2.675,-1.7) {\small Running part \\\small (vehicle spends the free flow \\\small time before joining the exit queue)};
\node at (6.9,-1.25) {\small Exit queue};
\end{tikzpicture}
\caption{Spatial-queuing structure of a link.}
\label{fig:link}
\end{figure}
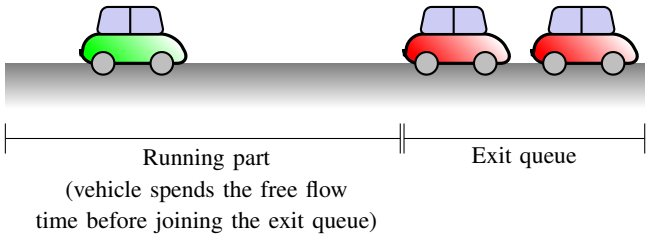

The following is a brief description of how the simulator routes the traffic through the network:
\begin{itemize}
\item Once a vehicle enters a link, it must spend the link's \textit{free flow travel time} before joining the link's exit queue.
\item  If the physical length of the exit queue reaches the upstream end of the link, then no more vehicles can enter the link (causing \textit{spillback}).
\item We assume that the links have a \textit{flow capacity} of 1900 vehicles per hour per lane. The number of vehicles that can enter/exit a link at a given time step is probabilistic and depends on its inflow/outflow capacity. 
\item These capacities are updated with the removal or addition of a vehicle, ensuring that the vehicle at the front of the exit queue has a higher probability of moving to its next link.
\item The node model determines the movement of the vehicles at the intersections. 
\item At each time step, the node model moves vehicles from the front of each link's exit queue to their next links, provided that:
\begin{itemize}
\item Doing so doesn't exceed the length capacity and the inflow capacity of the next link, nor the outflow capacity of the current link.
\item Vehicles do not have conflicting turn-directions with respect to a primary moving vehicle (e.g., vehicles with perpendicular directions cannot move).
\end{itemize}
\item All intersections are non-signalized, and vehicles entering an intersection have equal priority, except in roundabouts, where they have higher priority.
\end{itemize}

Electric vehicles requiring a charge visit charging stations before completing their trips. We model each charging station as a 3-stage process, consisting of an entrance queue, charging ports, and an exit queue. When a vehicle enters a charging station, it joins the end of the entrance queue. At each timestep, if there is an available charging port, the vehicle at the front of the entrance queue moves to that port, and the queue is adjusted accordingly. The charge levels of the vehicles at the ports increase by a certain amount at each timestep. Once a vehicle reaches its target charge, it proceeds to the end of the exit queue, if there is space. Vehicles then return to the traffic flow on a first-in, first-out basis. The entrance and exit of the charging station follow the node model, adhering to standard traffic rules. Figure~\ref{fig:charging_station} illustrates the charging station operations that we just described. 

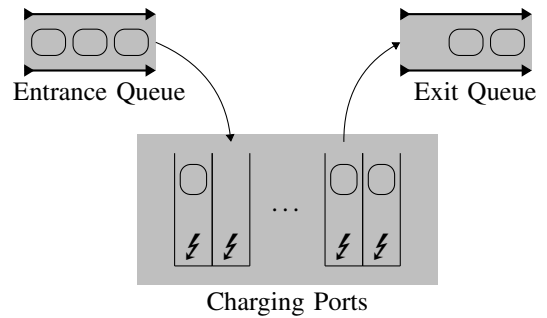
\begin{figure}[h]
\centering
\begin{tikzpicture}[scale=0.5]
    \fill [gray!50] (0,0.5) rectangle (8,4.5);

    \draw (1,1) to (1,4);
    \draw (1,1) to (2,1);
    \draw (2,1) to (2,4);
    \draw[rounded corners] (1.15,2.9) rectangle (1.85,3.7) {};
    \node at (1.5,1.5) {\Large\Lightning};
    \draw (2,1) to (2,4);
    \draw (2,1) to (3,1);
    \draw (3,1) to (3,4);
    \node at (2.5,1.5) {\Large\Lightning};
    \node at (4,2.5) {\ldots};
    \draw (5,1) to (5,4);
    \draw (5,1) to (6,1);
    \draw (6,1) to (6,4);
    \draw[rounded corners] (5.15,2.9) rectangle (5.85,3.7) {};
    \node at (5.5,1.5) {\Large\Lightning};
    \draw (6,1) to (6,4);
    \draw (6,1) to (7,1);
    \draw (7,1) to (7,4);
    \draw[rounded corners] (6.15,2.9) rectangle (6.85,3.7) {};
    \node at (6.5,1.5) {\Large\Lightning};
    \node at (4,0) {Charging Ports};

    \fill[gray!50] (-3,6.2) rectangle (0.5,7.7);
    \draw[thick,>->] (-3,6.2) -- (0.5,6.2);
    \draw[thick,>->] (-3,7.7) -- (0.5,7.7);
    \draw[rounded corners] (-1.9, 6.6) rectangle (-2.8, 7.3) {};
    \draw[rounded corners] (-0.8, 6.6) rectangle (-1.7, 7.3) {};
    \draw[rounded corners] (0.3, 6.6) rectangle (-0.6, 7.3) {};
    \node at (-1,5.6) {Entrance Queue};
    \draw[->] (0.5,6.95) to[bend left] (2.5,4.3);

    \fill[gray!50] (7,6.2) rectangle (10.5,7.7);
    \draw[thick,>->] (7,6.2) -- (10.5,6.2);
    \draw[thick,>->] (7,7.7) -- (10.5,7.7);
    \draw[rounded corners] (9.2, 6.6) rectangle (8.3, 7.3) {};
    \draw[rounded corners] (10.3, 6.6) rectangle (9.4, 7.3) {};
    \node at (9,5.6) {Exit Queue};
    \draw[->] (5.5,4.3) to[bend left] (7,6.95);

\end{tikzpicture}
\caption{Illustration of a charging station.}
\label{fig:charging_station}
\end{figure}

For more information on the simulator (excluding the charging behavior and stations) we refer to \cite{Sim_Package_Git,Sim_Package_Report,zhao2019agent}.

\subsection{Equilibrium in the Queuing-Based Model} \label{sec:queue-based-NE}
\noindent
As in the congestion game described in Section~\ref{sec:ev-as-congestion-game}, we consider a Nash equilibrium as the traffic operating point of the queuing-based model. Recall that the system is at a Nash equilibrium if no agent has an incentive to unilaterally deviate from its strategy. This means that, in the queuing-based model, any route assignment in which all utilized routes have minimal delay corresponds to a Nash equilibrium.

A key question, then, is how to compute a Nash equilibrium of the queuing-based model. Since the model is not guaranteed to be a potential game, the theoretical convergence guarantees for the “better-response” heuristics in \cite{monderer-shapley-potential-games,sandholm-potential-games,kara-et-al-erlang} do not apply. Moreover, existing work on dynamic equilibrium in queue-based models focuses on the non-atomic setting (e.g., \cite{Skutella2009flows}), whereas our focus is on computing equilibria for an atomic simulation. Nevertheless, we can still implement a better-response heuristic alongside an equilibrium check to likely steer the system toward a Nash equilibrium. In this approach, we iteratively move a driver from a route with the highest delay to a route with the lowest delay. After each iteration, we check whether all utilized routes have approximately minimal delays. Whenever this condition is satisfied, the system is close to a Nash equilibrium. In our upcoming experiments (see Section~\ref{sec:queue-experiment-results}), we compute an approximate equilibrium route assignment using a better-response heuristic -- this is an approximate equilibrium. Across all of our simulations, the resulting route assignments consistently satisfied the approximate Nash equilibrium condition within a finite number of iterations.

\section{Optimal EV charger placement} 
\noindent
Let $S\subseteq V_c$ be a selection of nodes for placing the EV stations and $x^*(S)$ be the Nash equilibrium flow for that selection. Let $c_d(x)$ represent the cumulative delay given the link flows $x$. The optimal EV placement problem can be defined as follows.
\begin{align}
\label{eq:Outeropt}
   \textrm{minimize } c_d(x^*(S)) 
   \\
   \textrm{subject to }  S \subseteq V_c, |S| = n_s      
\end{align}

In words, minimize the total delay experienced by the users over all of possible selections of $n_s$ EV charging station locations from the candidate set. 

Optimizing a set function becomes especially challenging when, as in our case, the function is neither submodular nor monotonic. To address this, we propose a greedy approach, outlined in Algorithm \ref{alg:greedyEV}, which provides an approximate solution to the optimization problem.

Our approach consists of adding EV nodes incrementally, one at a time. At each step, we identify the location that yields the greatest improvement in the objective at the equilibrium point, and add it to the set of nodes. The resulting solution is then locally refined using single swaps: at each step, we pick an unselected candidate location and consider exchanging it with a currently selected one. We apply the swap that yields the largest reduction in congestion, if any. This procedure is repeated for a fixed number of iterations, or until no further improvement is observed.  
We should mention that this is the scalable approach, but for smaller networks, alternative methods can be explored.

\begin{algorithm}[htb]
\SetKw{Function}{Function}
\SetKwBlock{FnBody}{is}{end}
\caption{Greedy EV Station Placement}
\label{alg:greedyEV}
\KwIn{Set of candidate locations $V$, required number of EV stations $n_s$, road network $(V,E)$, 
agent types/strategy sets/utility functions.} 
\KwOut{Selection of EV stations $V_o \subseteq V$, $|V_o|=n_s$.}

$V_o \gets \emptyset $\;
\While{$|V_o| < n_{s}$} {
\For{$e \in V \setminus V_o$} { 
$x^*(V_o \cup \{e\}) \gets$ NE allocation at $V_o \cup \{e\}$ 
}
$e^* \gets $ minimizer of $c_d(x^*(V_o \cup \{e\}))$ for $e \in V \setminus V_o$\;
$V_o \gets V_o \cup \{e^*\}$
}
Optional: Single-swap refinement\;
\Return{$V_o$}\;
\end{algorithm}

\subsection{Greedy Selection does not guarantee optimality}
\noindent
As we show in the Experiments and Results section, the greedy placement is likely to be very close to the optimal placement. Unfortunately, this method cannot guarantee optimal placement in theory, as we can see in the following example.

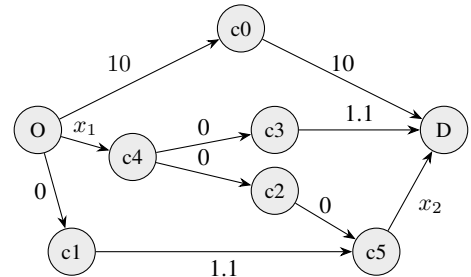
\begin{figure}[hbt]
    \centering
    \begin{tikzpicture}[scale=0.9, every node/.style={transform shape}]

\node[big] (O) at (-2,0) {O};
\node[node] (c0) at (1,1.5) {c0};
\node[node] (c4) at (-0.6,-0.4) {c4};
\node[node] (c3) at (1.5,0) {c3};
\node[node] (c2) at (1.5,-0.9) {c2};
\node[node] (c5) at (3,-1.8) {c5};
\node[node] (c1) at (-1.5,-1.8) {c1};
\node[big] (D) at (4,0) {D};

\draw[edge] (O) -- (c0) node[midway, above left] {$10$};
\draw[edge] (c0) -- (D) node[midway, above] {10};

\draw[edge] (O) -- (c4) node[midway, above] {$x_1$};
\draw[edge] (c4) -- (c3) node[midway, above] {0};
\draw[edge] (c3) -- (D) node[midway, above] {1.1};

\draw[edge] (c4) -- (c2) node[midway, above] {0};
\draw[edge] (c2) -- (c5) node[midway, above] {0};
\draw[edge] (c5) -- (D) node[midway, below right] {$x_2$};

\draw[edge] (O) -- (c1) node[midway, left] {0};
\draw[edge] (c1) -- (c5) node[midway, below] {1.1};

\end{tikzpicture}
    \caption{A Network with three possible EV charging spots given as $c_1,c_2,c_3$. The default charging station is $c_0$, and the problem is to pick two additional spaces to add to that one. The delays are noted on the corresponding links.}
    \label{fig:non-optimal}
\end{figure}

In figure \ref{fig:non-optimal}, we are given a single origin-destination pair, represented as $O, D$. Initially the only available charging station is $c_0$. We only have type $F_2$ drivers, and the total amount of flow between the origin and destination is $1$. The problem is to pick two places out of $\{c_1,c_2,c_3\}$.
The greedy EV station placement will pick $c_2$ first, then add one of $c_1$ or $c_3$ (these two have the same delay, so it will pick one at random), which leads to all flow following the path $O \rightarrow c_2 \rightarrow D$, hence $x_1=x_2=1$ and the resulting total delay will be $2$. 
The optimal solution, on the other hand, is the selection $c_1,c_3$ which leads to the flow equally divided between the paths $O \rightarrow c_3 \rightarrow D$ and $O \rightarrow c_1 \rightarrow D$, and the total delay will be $1.6$. 
\section{Methodology}
\label{sec:methodology}
\noindent

\begin{figure*}
    \centering
    \includegraphics[width=\textwidth]{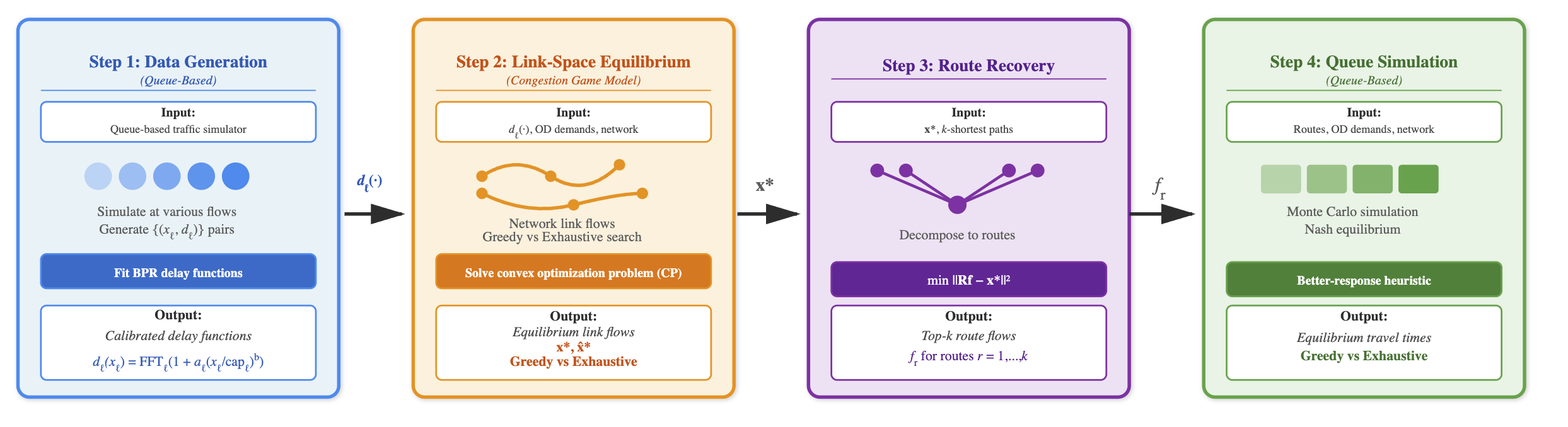}
    \caption{Four-step methodology for optimal EV charger placement. Queue-based simulation data (Step 1) calibrates BPR delay functions that feed into congestion game optimization (Step 2) to compute equilibrium link flows. These flows are decomposed into explicit routes via least-squares recovery (Step 3) and fed into queuing-based approximate equilibrium simulation (Step 4). The framework compares greedy and exhaustive search strategies in both models.}
    \label{fig:method}
\end{figure*}
\subsection{Multi-step optimization framework}
\noindent
To enable a fair, model-consistent comparison between the congestion-game formulation and the queue-based traffic simulator, we adopt a unified, multi-step workflow. The guiding objectives are:
\begin{enumerate}[label=(\roman*)]
    \item \textbf{Realistic link delays}: infer per-link delay parameters of the congestion game using data obtained from the queuing-based traffic simulator, yielding calibrated functions $d_\ell(\cdot)$ that reflect simulated traffic dynamics.
    \item \textbf{All-path coverage without route explosion}: avoid restricting to a small, possibly biased route set. Instead, we solve the inner problem(finding the NE) in link-space, which implicitly accounts for all possible routes $s$ in the route set, and only recover route flows afterward. This extends our prior work \cite{Sonmez2024Optimal} that optimized over a selected subset of paths.
    \item \textbf{Route-flow realization for simulation}: obtain a set of route flows that best explain the optimized link flows so they can be executed in the queue-based simulator, which requires explicit route assignments.
\end{enumerate}

Our method proceeds as follows.
\begin{enumerate}[label=\arabic*.]
    \item \textbf{Queue-based data generation and link parameter fitting.} Using the queue-based model, we simulate each physical link over a range of aggregate flows to obtain pairs $\{(x_\ell,\,\bar d_\ell)\}$, where $x_\ell$ is the measured flow and $\bar d_\ell$ is the average link travel time under stationary operation. For each link $\ell$, we estimate parameters $\hat\theta_\ell$ of the delay law in \eqref{fit_bpr} by least squares on the simulated dataset, yielding calibrated $d_\ell(\cdot;\hat\theta_\ell)$ for all links.
    \item \textbf{Equilibrium in link-space.} We solve the convex inner problem \eqref{eq:EPOPT} with delay functions $d_\ell(\cdot;\hat\theta_\ell)$ to obtain the equilibrium link flows $\mathbf{x}^*$ and charger throughputs $\hat{\mathbf{x}}^*$ for a given charger set $S$. Because the decision variables are link flows subject to flow-conservation, every feasible path is implicitly represented.
    \item \textbf{Route-flow recovery.} From $\mathbf{x}^*$, we recover route flows by solving the constrained least-squares program described in Section~\ref{sec:route-recovery}. The routing matrix $\mathbf{R}$ is built from a large library of $k$-shortest paths (non-charging and charging compositions). We then rank routes by $f_r$ and study top-$k$ coverage to select a subset that reliably reproduces $\mathbf{x}^*$.
    \item \textbf{Queue-based simulation with recovered routes.} Using the recovered route library and the corresponding flows as initialization, we run the queue-based simulator to its stationary regime, obtaining the simulator-side equilibrium and performance metrics. This enables fair comparison against the congestion-game equilibrium under the same calibrated link laws and OD demands.
\end{enumerate}

\subsection{Queue-based data generation and link parameter fitting}
\noindent
To ensure consistency between the congestion game model and the queue-based simulator, we first generate realistic traffic data using the queue-based model and then calibrate delay functions that represent the simulated traffic dynamics.

\subsubsection{Data generation process}
\noindent
Using the queue-based simulator, we simulate each physical link over a range of aggregate flows to obtain pairs $\{(x_\ell,\,\bar d_\ell)\}$, where $x_\ell$ is the measured flow and $\bar d_\ell$ is the average link travel time under stationary operation. This data generation process reflects realistic traffic dynamics that arbitrary delay functions may fail to represent.

\subsubsection{Link parameter fitting}
\noindent For each link $\ell$, we estimate parameters $\hat\theta_\ell$ of the delay by least squares on the simulated dataset. We fit the BPR (Bureau of Public Roads) model:
\begin{equation}\label{fit_bpr}
d_l(x_l) = FFT_l\left(1 + a_l \cdot \left(\frac{x_l}{cap_l}\right)^{b_l}\right)
\end{equation}
where $FFT_l$ is the free-flow travel time (travel time at zero congestion) for link $l$, $cap_l$ is the link capacity (maximum flow in vehicles per hour), and $a_l, b_l > 0$ are calibrated parameters that control the rate and curvature of delay increase with congestion. All four parameters $\{FFT_l, cap_l, a_l, b_l\}$ are estimated by nonlinear least-squares fitting to the empirical link delay data collected from the queue-based simulations.

\subsection{Equilibrium in link-space}
\noindent
In this step, we solve the congestion game optimization problem to find the Nash equilibrium flows directly in link-space. This approach avoids the computational explosion that would occur if we explicitly enumerated all possible routes, while still ensuring that every feasible path is implicitly considered.

The equilibrium solution is obtained by solving the convex optimization problem \eqref{eq:EPOPT} with the calibrated delay functions $d_\ell(\cdot;\hat\theta_\ell)$ from the previous step. This yields the equilibrium link flows $\mathbf{x}^*$ and charger throughputs $\hat{\mathbf{x}}^*$ for a given charger set $S$. 
Since the decision variables are link flows with flow-conservation constraints, all feasible paths are represented.

The key advantage of this link-space approach is computational tractability: rather than considering the exponentially large set of all possible routes, we work directly with link flows, which scale linearly with the number of links in the network. This enables us to solve large-scale problems that would be intractable using traditional route-based formulations.
\subsection{Route-flow recovery}\label{sec:route-recovery}
\noindent
In a standard network flow problem without additional constraints, any feasible flow can be exactly decomposed into path flows using the classical flow decomposition theorem. In our setting, due to computational constraints, we approximate this by considering only a subset of all possible paths as follows.

For each OD pair, we generate candidate routes using k-shortest path algorithms. For non-charging routes, $k_{od}$ shortest paths from origin to destination. For charging routes, combinations of $k_{oc}$ shortest paths from origin to charger and $k_{cd}$ shortest paths from charger to destination

The route recovery problem is formulated as a constrained least-squares optimization:
\begin{align}
\min_{\mathbf{f}} \quad & \|\mathbf{R}\mathbf{f} - \mathbf{x}^*\|_2^2 \notag \\
\text{s.t.} \quad
& \sum_{r \in \mathcal{R}_{i,t}} f_r = q^{(i,t)}, \quad \forall i, t \tag{\textit{OD demand constraints}}\\
& \sum_{r \in \mathcal{R}_{i,t,c}} f_r = q^{(i,t,c)}, \quad \forall i, t, c \tag{\textit{Charger flow constraints}}\\
& f_r \geq 0, \quad \forall r \tag{\textit{Non-negativity}}
\end{align}

where $\mathbf{R}$ is the routing matrix (1 if route $r$ uses link $l$, 0 otherwise), $\mathbf{f}$ is the vector of route flows, $\mathbf{x}^*$ are the optimal link flows from Phase 1, and $\mathcal{R}_{i,t}$ and $\mathcal{R}_{i,t,c}$ denote the sets of routes for OD pair $i$, type $t$, and charger $c$ respectively.

\subsection{Queue-based simulation with recovered routes}
\label{sec:queue-meth}
\noindent
The final step in our methodology is to evaluate the greedy placement algorithm using the recovered routes from Section~\ref{sec:route-recovery} and the queuing-based simulation tool.

\subsubsection{Nash Equilibrium Route Assignments} \label{sec:ne-route-assignments}
\noindent
We set the departure time of all agents to 0 and, similar to the congestion game framework, we take a Nash equilibrium as the operating point of traffic. Therefore, we need to compute an approximate equilibrium of the queue-based model to assign the routes.

To approximate this equilibrium, we use the better-response heuristic outlined in Section~\ref{sec:queue-based-NE}. In each iteration, we shift one driver from the route with the highest delay to the route with the lowest delay, run $N$ independent simulations under the updated assignments, and compute the Monte Carlo average of agent-wise mean travel times for each route (unused routes' travel times are set to their free flow times). This process is repeated until the link travel times approximately satisfy the Nash equilibrium condition--that is, the travel times on all used links are nearly identical and minimal across all feasible routes. We note that the simulator introduces randomness primarily to capture flow capacity constraints (see Section~\ref{sec:queue-based-model}). This motivates our use of Monte Carlo averaging to estimate expected travel times.

In realistic road networks, the number of routes connecting an origin–destination pair can be prohibitively large. To ensure tractability, we restrict attention to the top-$k$ routes obtained in Section~\ref{sec:route-recovery} (using the congestion game framework). We select $k$ by analyzing coverage and error metrics as functions of the number of routes (see Figure~\ref{fig:k_routes_analysis}). For example, Figure~\ref{fig:k_routes_analysis} shows that $k=16$ achieves near-complete coverage while keeping the route set manageable. For the initial assignment, we use the Nash equilibrium flows derived in Section~\ref{sec:route-recovery}.

\subsubsection{Evaluating greedy placement}\label{sec:eval-greedy-methodology}
\noindent
We evaluate the performance of the greedy charger placement algorithm by comparing it to the results of exhaustive search. For each approach, given a set of EV charger locations, we compute the Monte Carlo average of the total travel time experienced by all agents across $M$ independent simulations. In each simulation, routes are assigned according to the Nash equilibrium derived using the methodology described above.

Exhaustive search evaluates all possible charger placements, whereas the greedy algorithm considers only the subset of placements determined by Algorithm~\ref{alg:greedyEV}. Because exhaustive search explores the full solution space, it identifies the placement that minimizes total travel time (and, thus, congestion). We assess the greedy algorithm by determining whether it can find the optimal placement identified by exhaustive search.

\section{Experiments and Results}
\subsection{Experimental setup}
\noindent
We evaluate the overall workflow on a real-world roadway network in College Park, MD, USA, through simulations involving 180 vehicles, of which 120 require charging. The network, shown in Figure~\ref{fig:network}, was extracted from OpenStreetMap. Our experimental pipeline uses calibrated link-delay models (Section~\ref{sec:methodology}), solves the equilibrium in link-space, recovers a compact set of routes, and simulates those routes in the queue-based model.

\begin{figure}[htb!]
    \centering
    \includegraphics[width=0.48\linewidth]{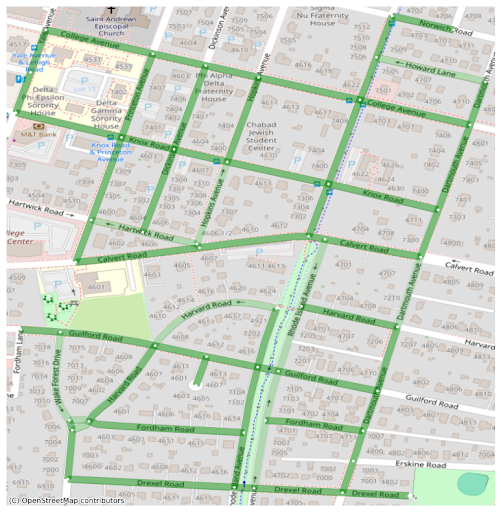}
    \includegraphics[width=0.48\linewidth]{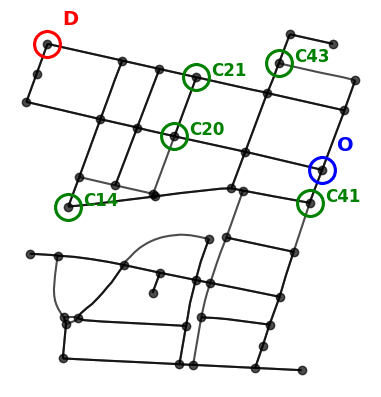}
    \caption{Study network extracted from OpenStreetMap (College Park, MD, USA). Left: Geographic overview of the network. Right: Schematic representation with candidate charging locations (green circles), origin node `O', and destination node `D'. The network has 48 nodes and 123 directed links.}
    \label{fig:network}
\end{figure}

\subsection{Link-delay calibration results}
\noindent
Using the queue-based simulator, we generated link-level datasets and fitted BPR-style delay functions. Figure~\ref{fig:link-fits} shows representative fits: the top row exhibits high $R^2$ (high agreement), while the bottom row illustrates cases with lower $R^2$ where the simple BPR law underfits complex dynamics.

\begin{figure}[htb]
    \centering
    \includegraphics[width=0.4925\linewidth]{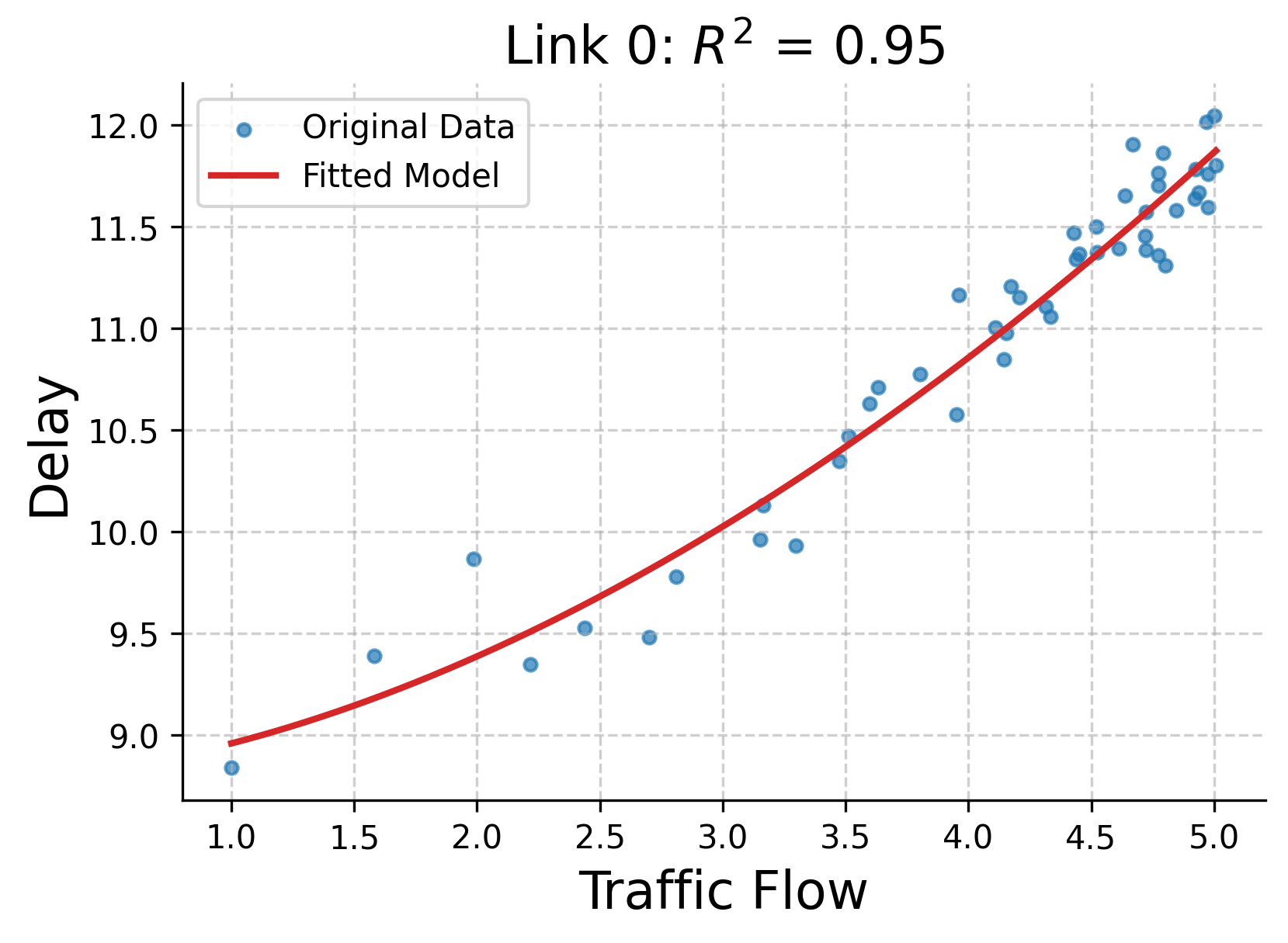}
    \includegraphics[width=0.4925\linewidth]{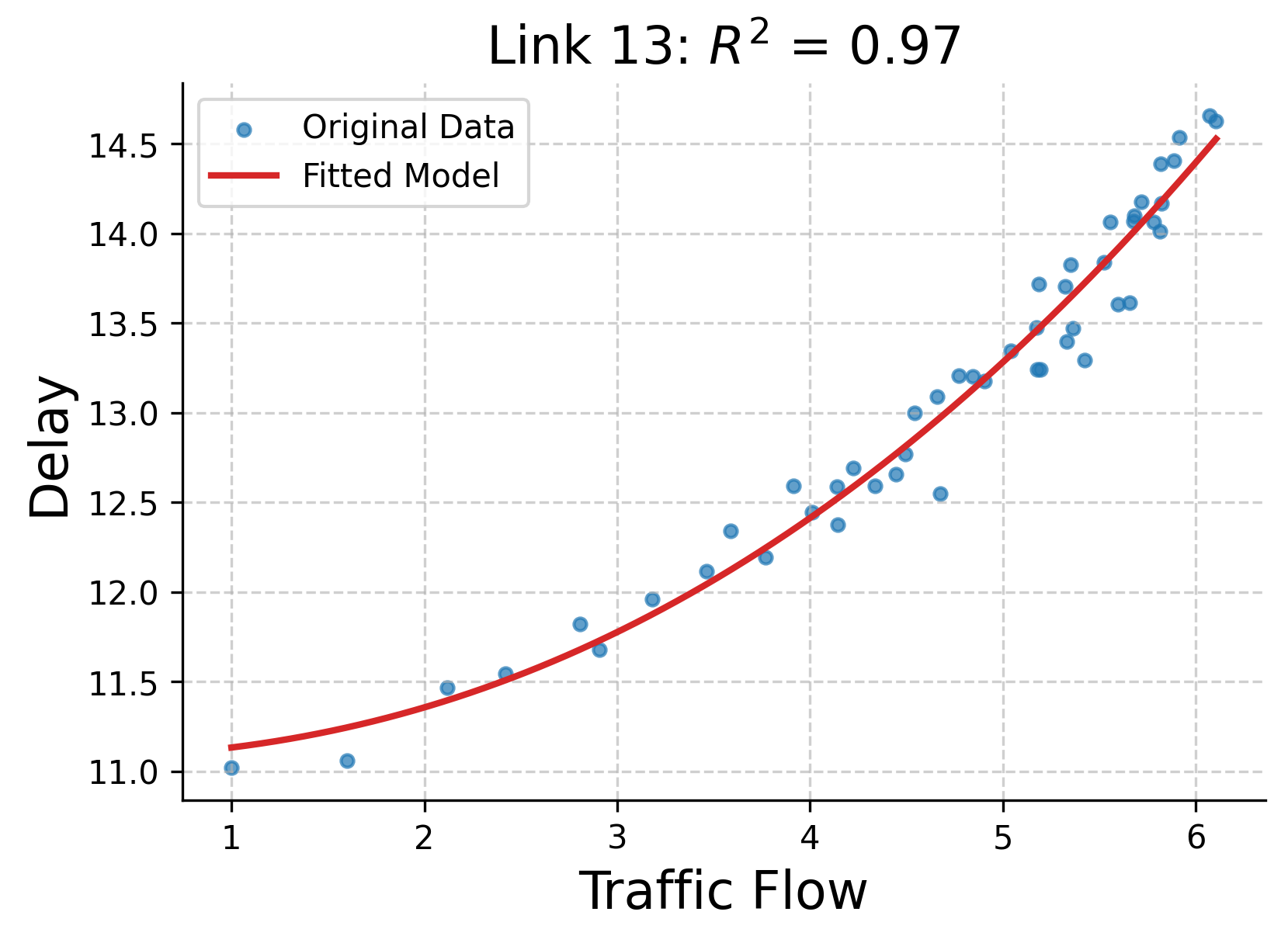}
    \includegraphics[width=0.4925\linewidth]{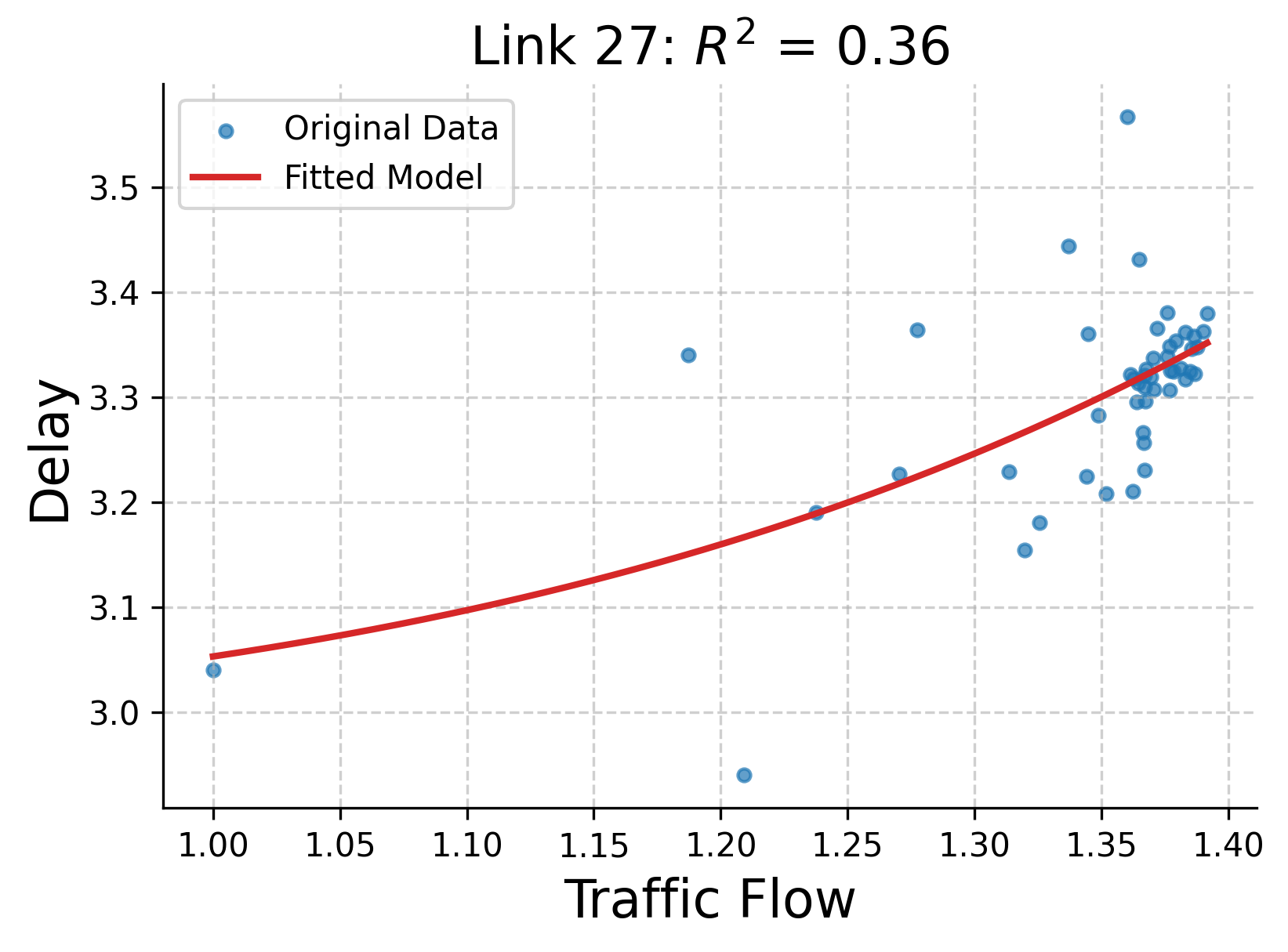}
    \includegraphics[width=0.4925\linewidth]{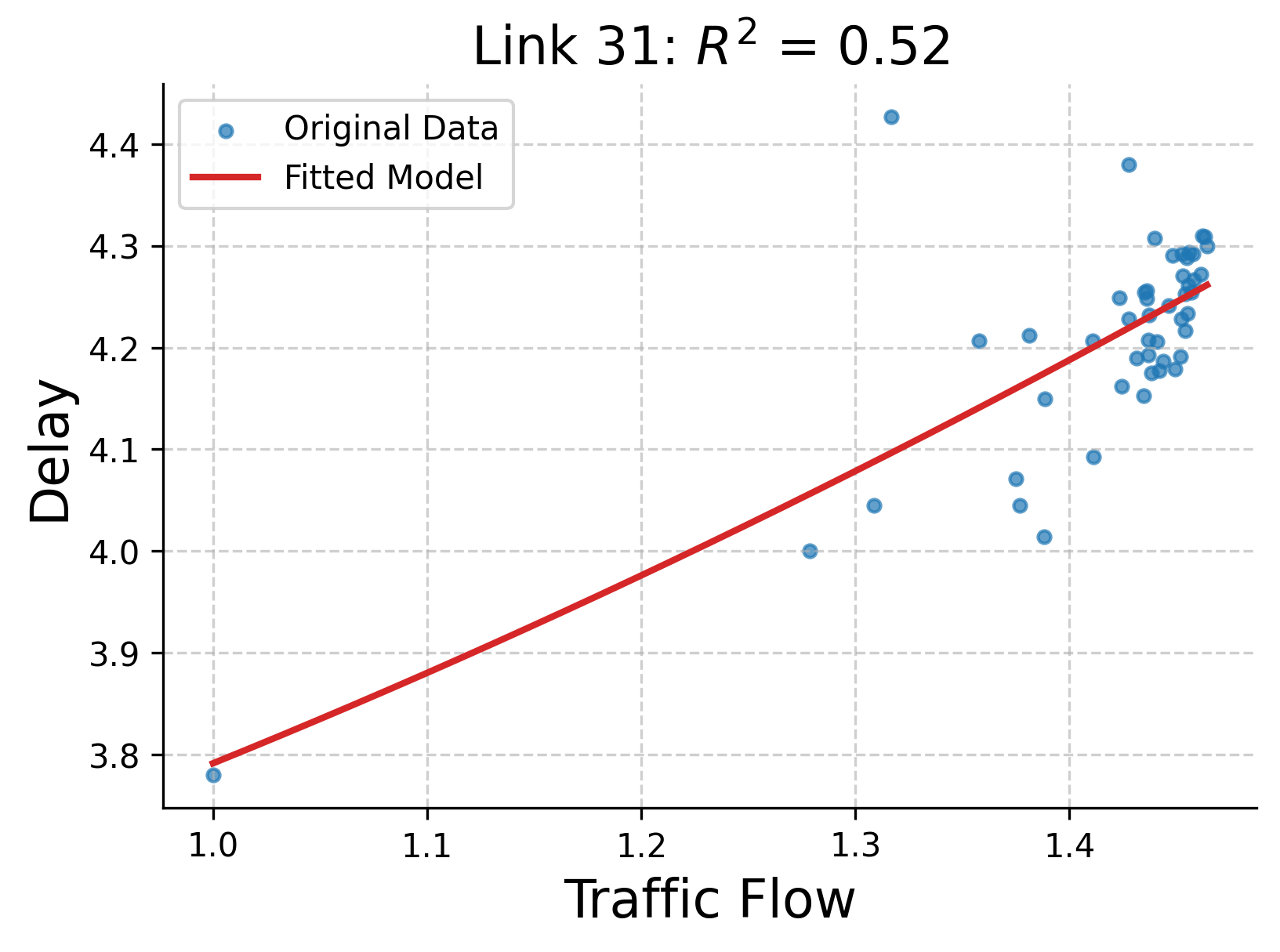}
    \caption{Examples from link-parameter fits. Top: high $R^2$; Bottom: lower $R^2$.}
    \label{fig:link-fits}
\end{figure}

\subsection{Route-flow recovery evaluation}\label{subsec:route-flow-recovery}
\noindent
We assess how well a compact set of routes can reproduce the optimal link flows from the equilibrium solver. For a given number of routes, we: (i) sort recovered routes by flow magnitude, (ii) select the top-$k$ from sorted route flows and reconstruct link flows, and (iii) compute coverage, MAE, RMSE, and correlation versus the original link flows.

\begin{figure}[h!]
    \centering
    \includegraphics[width=0.5\textwidth]{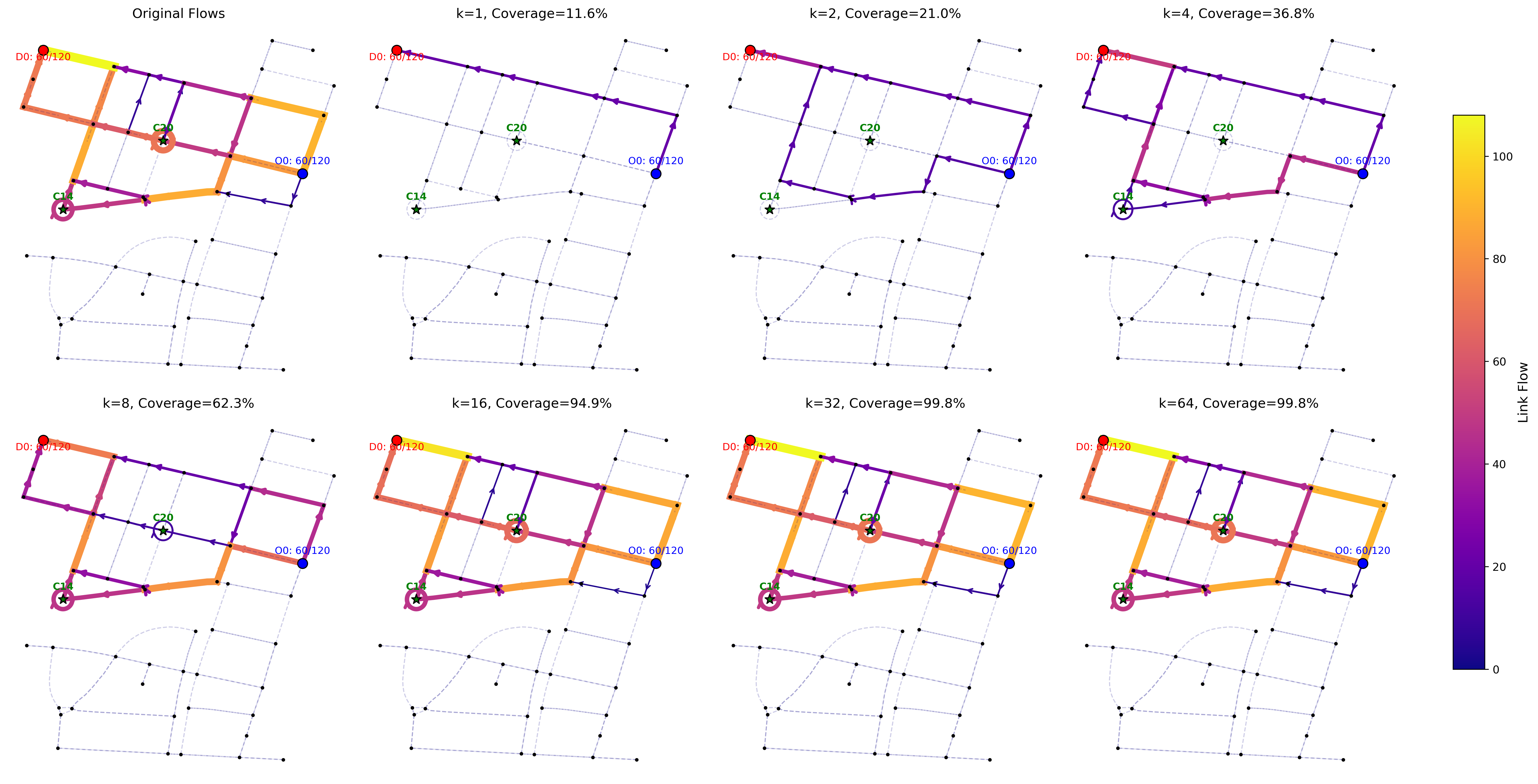}
    \caption{Progressive flow reconstruction as the number of routes increases.}
    \label{fig:flow_reconstructions}
\end{figure}

\begin{figure}[htb]
    \centering
    \includegraphics[width=0.46\textwidth]{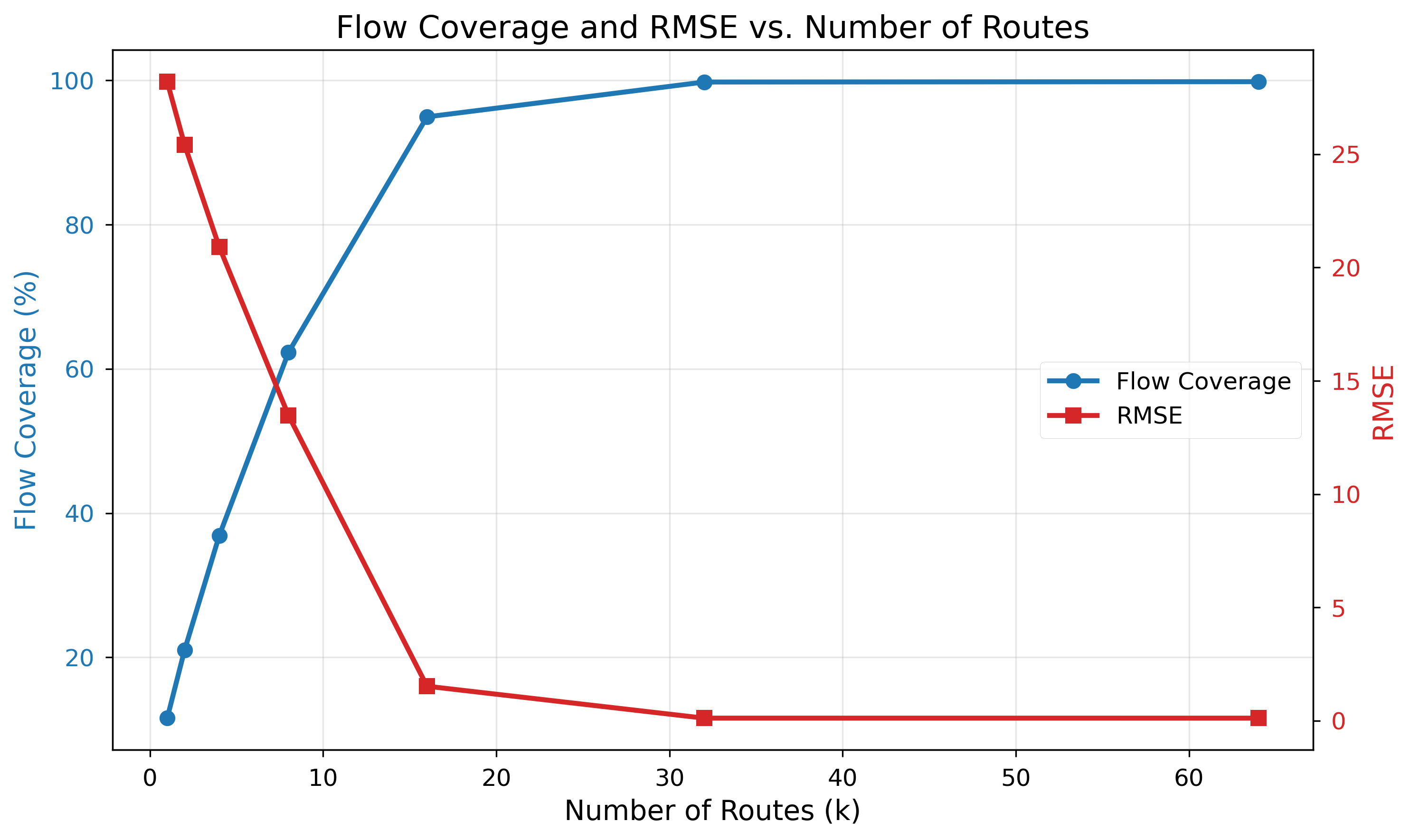}
    \caption{Coverage and RMSE metrics by the number of routes.}
    \label{fig:k_routes_analysis}
\end{figure}

The results reveal a clear trade-off between route set size and accuracy, as can be seen in the Figure \ref{fig:flow_reconstructions}. Increasing $k$ improves coverage and reduces error across all metrics, with diminishing returns near $k=16$--$32$. With 8 routes, we achieve 62.3\% coverage (RMSE 13.5); with 16 routes, coverage rises to 94.9\% (RMSE 1.52); with 32 routes, we reach 99.8\% coverage with negligible error (RMSE 0.12). Thus, a small set of high-flow routes explains most network flow while remaining tractable for simulation.

\subsection{Queue-based simulation with recovered routes} \label{sec:queue-experiment-results}
\noindent
In the final stage of our experiments, we evaluate the greedy charger placement algorithm on the queuing-based model, using the methodology in Section~\ref{sec:queue-meth}.

We initialize the queuing-based simulator with the top-$16$ recovered routes, where we based the choice $k=16$ on the coverage metrics obtained in Section~\ref{subsec:route-flow-recovery} (displayed in Figure~\ref{fig:k_routes_analysis}). This choice ensures substantial coverage ($94.9\%$), while increasing $k$ to 32 yields negligible improvement.

Now, for every possible combination of charger locations, we compute a Nash equilibrium route assignment. In all simulations, the better-response heuristic converged near a Nash equilibrium, where we used the following inequality as the convergence condition:
\begin{align*}
&\max_{s,\tilde{s}\in R}\left|\sum_{l\in s}d_l(x_l) - \sum_{k\in \tilde{s}}d_k(x_k)\right| \leq \alpha \min_{s \in S} \left( \sum_{l\in s}d_l(x_l) \right).
\end{align*}
Here, $\alpha=0.01$ is a constant, $S$ is the set of all available routes ($|S|=16$), $R$ is the subset of routes used by at least one agent, and $d_l(x_l)$ is the agent-wise average travel time of link $l$ under the link flow $x_l$. Thus, this inequality means that a route assignment is near a Nash equilibrium if the maximum average travel time discrepancy among used routes is less than 1\% of the minimum possible travel time. Figure~\ref{fig:nash_flows} displays the heatmap of the resulting Nash equilibrium flow for charging station locations C20 and C43. For comparison, we place it side-by-side with the corresponding Nash equilibrium flows from the congestion game model. This comparison illustrates that both models use essentially the same set of links, and the least-used links coincide across models. However, several links that carry only moderate flow in the congestion game are used much more heavily in the queuing-based model. Although we calibrated the congestion game to align closely with the queuing-based model, differences remain because the BPR function cannot capture certain behaviors inherent to the queuing-based formulation.

\begin{figure}[h]
\centering
{\fontfamily{DejaVuSans-TLF}\selectfont
\begin{minipage}{0.48\linewidth}
    \centering
    {\tiny{Congestion Game}}\\[3pt]
    \includegraphics[width=\linewidth]{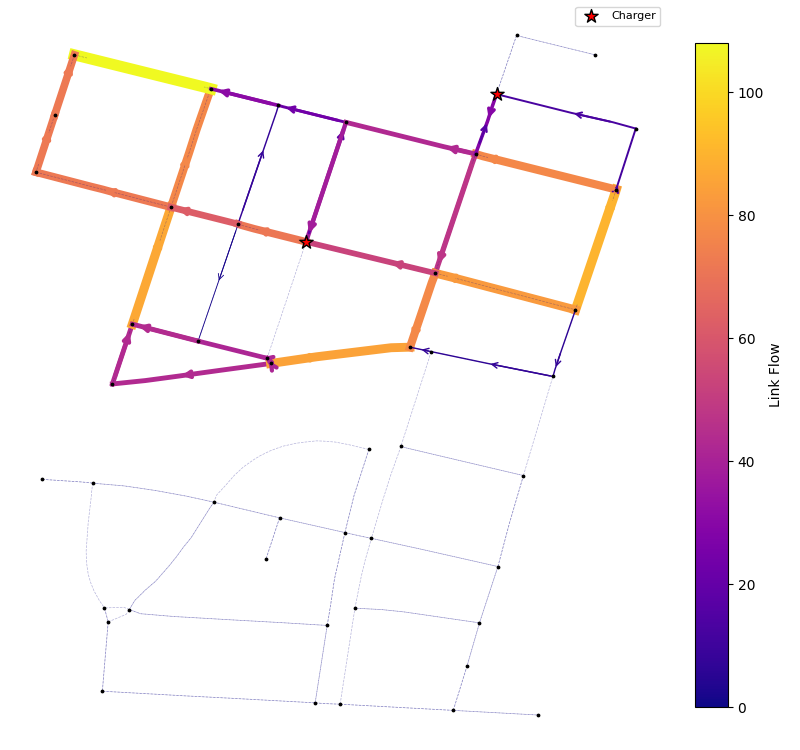}
\end{minipage}
\hfill
\begin{minipage}{0.48\linewidth}
    \centering
    {\tiny{Queuing-Based Model}}\\[3pt]
    \includegraphics[width=\linewidth]{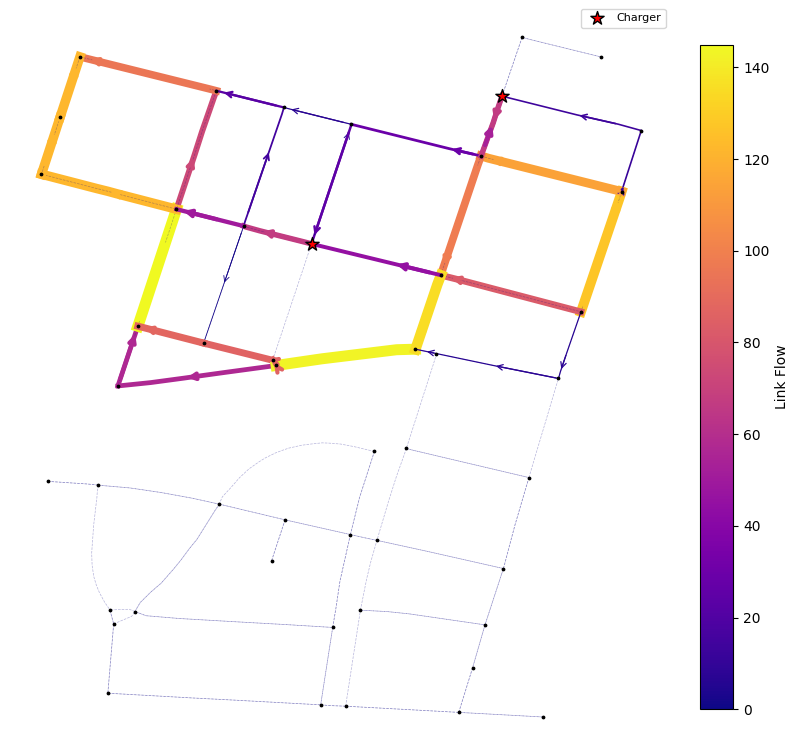}
\end{minipage}
}
\caption{Heatmaps of the equilibrium flows produced by the congestion game (Nash equilibrium) and the queuing-based model (approximate equilibrium) for charging station locations C20 and C43.}
\label{fig:nash_flows}
\end{figure}

We now evaluate the performance of the greedy placement method. For each set of charger locations, we assume that traffic follows the approximate Nash equilibria computed above and apply the analysis in Section~\ref{sec:eval-greedy-methodology}. Within each simulation setup (i.e., each charger location combination), we estimate the expected total travel time by running 100 independent Monte Carlo simulations and averaging their results. Figure~\ref{fig:greedy_vs_exhaustive} presents the outcomes of the exhaustive search across all charger combinations, along with the route with the smallest delay identified by the greedy placement algorithm. As shown in the figure, \textit{both exhaustive search and greedy placement label the same combination of charger locations as optimal}, confirming the validity of the greedy placement approach. Notably, the greedy algorithm identifies node C21 as the best initial charger location (which is consistent with the results from exhaustive search), and then evaluates additional locations in combination only with node C21.

Figure~\ref{fig:greedy_vs_exhaustive} also displays the normalized travel times from the congestion game model. For this model, the output of exhaustive search and greedy placement coincide. However, the optimal location-tuple that the congestion game yields is different than the one from the queuing-based model. In fact, we see a notable difference between the ranking of sets of charger locations. This difference is not surprising, as the congestion game model is significantly different than the queuing-based model. Nonetheless, the performance of the greedy placement algorithm withstands this difference. 

\textit{Code and Data is available upon request.}

\begin{figure}
    \centering
    \includegraphics[width=\linewidth]{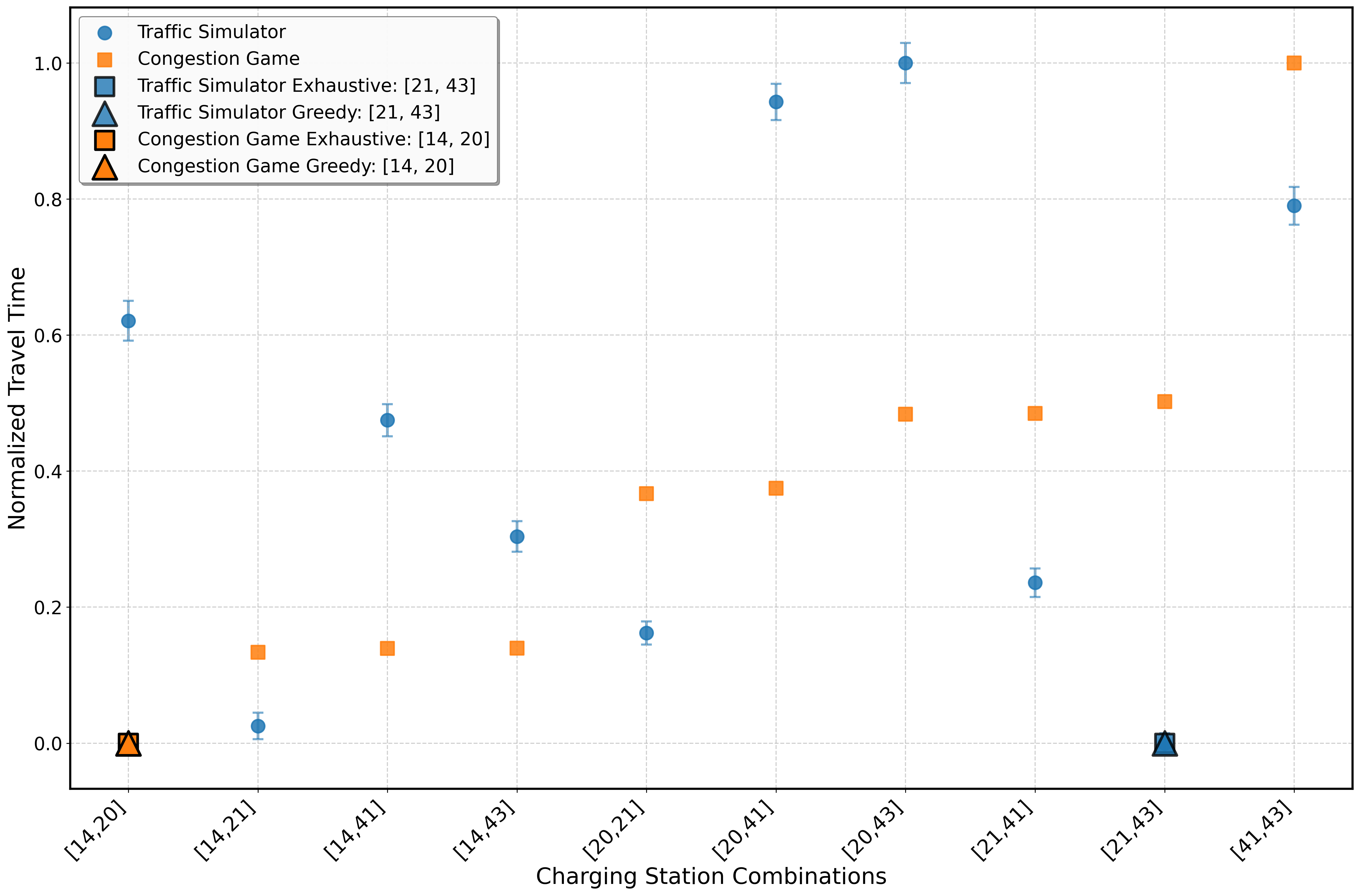}
    \caption{Queue-based and Congestion game model results showing normalized total travel time for each charger location combination. Bars indicate mean values from 100 Monte Carlo simulations; error bars show standard error. While optimal locations differ between models, greedy placement remains effective in both formulations.}   
    \label{fig:greedy_vs_exhaustive}
\end{figure}

\section{Conclusion} 
\noindent
In this paper, we formulated the EV charging station placement problem using two models: (i) a classical congestion game framework, and (ii) a queue-based model that is more realistic but more difficult to analyze.
We found that the congestion game model can yield equilibrium outcomes that differ substantially from those of the queue-based model. Although we also provide a counterexample showing that greedy selection does not guarantee optimality, our experiments indicate that it nevertheless produces near-optimal solutions in practice. Finally, both formulations are scalable: the congestion game model is directly scalable by construction, and for the queue-based model, we observe that despite the potentially large number of OD routes, only a small fraction of feasible routes is actually used in equilibrium, which keeps the computational need manageable.

\bibliographystyle{IEEEtran}
\bibliography{sample}

\end{document}